\documentclass{article}

     \usepackage[preprint]{neurips_2018}

\usepackage[utf8]{inputenc} % allow utf-8 input
\usepackage[T1]{fontenc}    % use 8-bit T1 fonts
\usepackage{hyperref}       % hyperlinks
\usepackage{url}            % simple URL typesetting
\usepackage{booktabs}       % professional-quality tables
\usepackage{graphicx} 
\usepackage{amsfonts}       % blackboard math symbols
\usepackage{nicefrac}       % compact symbols for 1/2, etc.
\usepackage{microtype}      % microtypography
\usepackage{natbib}
\usepackage{amsmath}
\usepackage{amstext}
\usepackage{algorithm}
\usepackage[noend]{algpseudocode}
\usepackage{amsthm}

\usepackage{xcolor}

\newtheorem{thm}{Theorem}
\newtheorem{defn}{Definition}

\newtheorem{question}{Question}

\newtheorem{ass}{Assumption}

\usepackage{mathtools}

\title{Robust Multi-agent Counterfactual Prediction}

\author{%
  Alexander Peysakhovich\thanks{Equal contribution, author order has been randomized.} \\
  Facebook AI Research
  \And
  Christian Kroer\footnotemark[1] \\
  Facebook Core Data Science
  \And
  Adam Lerer\footnotemark[1] \\
  Facebook AI Research
}

\begin{document}
% \nipsfinalcopy is no longer used

\maketitle

\begin{abstract}
We consider the problem of using logged data to make predictions about what would happen if we changed the `rules of the game' in a multi-agent system. This task is difficult because in many cases we observe actions individuals take but not their private information or their full reward functions. In addition, agents are strategic, so when the rules change, they will also change their actions. Existing methods (e.g. structural estimation, inverse reinforcement learning) make counterfactual predictions by constructing a model of the game, adding the assumption that agents' behavior comes from optimizing given some goals, and then inverting observed actions to learn agent's underlying utility function (a.k.a. type). Once the agent types are known, making counterfactual predictions amounts to solving for the equilibrium of the counterfactual environment. This approach imposes heavy assumptions such as rationality of the agents being observed, correctness of the analyst's model of the environment/parametric form of the agents' utility functions, and various other conditions to make point identification possible. We propose a method for analyzing the sensitivity of counterfactual conclusions to violations of these assumptions. We refer to this method as robust multi-agent counterfactual prediction (RMAC).  We apply our technique to investigating the robustness of counterfactual claims for classic environments in market design: auctions, school choice, and social choice. Importantly, we show RMAC can be used in regimes where point identification is impossible (e.g. those which have multiple equilibria or non-injective maps from type distributions to outcomes).
\end{abstract}

\section{Introduction}
All markets have rules and some rules work better than others \citep{roth1999redesign,roth2005kidney,abdulkadirouglu2005boston,klemperer2002really,porter2003combinatorial}. Figuring out which rules yield good outcomes is the bread and butter of market design, an interdisciplinary field focused on the engineering of effective rules \citep{roth2002economist}. Good market design is particularly important for businesses which make their livelihoods as platforms (e.g. internet ad auctions, ride sharing, dating sites). A key challenge for market designers is to observe an existing set of rules at work and make a counterfactual statement about how outcomes would change if the rules changed \citep{bottou2013counterfactual,athey2015machine}.

The multi-agent counterfactual question is difficult for two reasons. First, participants are strategic. An agent's optimal action can change due to changes in the rules of the game, and often, can change when other agents change what they are doing. Second, agents have private information that is not known to the designer so even knowledge of the rules, and ability to compute optimal actions, is insufficient to estimate counterfactual outcomes. For example, if we observed data from a series of first-price sealed bid auctions, we could not assume that agents would continue to bid the same way if we changed the auction format to second price with a reserve.

The technique of structural estimation deals with these issues by assuming that observed actions are coming from a multi-agent system that, through repetition or other forces, has come to equilibrium. Further, it is often assumed that once changes are made, the system will again equilibriate. This means the counterfactual question becomes asking about how equilibria change as we make design changes. A downside of the standard structural approach is that it requires strong assumptions that are not always completely true in practice. For example, this process requires assuming that agents are optimizing their utility given the behavior of others so that an analyst can infer underlying `taste' parameters from agent actions \citep{berry1995automobile,athey2010structural}. It is well known, however, that human decisions do not always obey the axioms of utility maximization \citep{camerer2011advances} and that both mistakes and biases can persist even when there is ample opportunity for learning \citep{erev1998predicting, fudenberg2016recency}. 

Our main contribution is to propose a method which allows an analyst to see how robust their counterfactual conclusions are to relaxations of the assumptions of rationality and correct specification of the model. We first show that the counterfactual estimation problem can be written as a game which we call a revelation game. When the standard assumptions are satisfied the revelation game has a unique equilibrium. Looking at the set of $\epsilon$-equilibria of the revelation game is equivalent to relaxing these assumptions.

To apply this idea in practice we need to solve for particular equilibria of the revelation game - the `worst' and `best' elements of the $\epsilon$-equilibrium set with respect to some evaluation function (e.g. revenue). These equilibria form the upper and lower bounds for our robust multi-agent counterfactual prediction (RMAC). Varying $\epsilon$ gives the analyst a measure of how confident they can be in their inferences as they relax how strictly the standard assumptions hold. In addition, RMAC can be applied even when standard assumptions about point identification do not hold (e.g. when there are multiple equilibria or when the data is consistent with multiple type distributions) to compute optimistic and pessimistic counterfactual predictions.

The RMAC bounds are different from standard uncertainty bounds (e.g. the standard error of a maximum likelihood estimator). Statistical uncertainty bounds (i.e. standard errors) reflect variance introduced by access to only finite data but still assume the underlying model is completely correct. On the other hand, our robustness bounds are intended to measure error that can come from the analyst using a model that is precisely incorrect but approximately true.

We show that computing the RMAC bounds exactly is a difficult problem as it is NP-hard even for $2$-player Bayesian games. We propose a first-order method based on fictitious play applied to the revelation game which we refer to as revelation game fictitious play (RFP) to compute the RMAC bounds. 

We apply RFP to generate RMAC in three domains of interest: auctions, matching, and social choice. In each of them we find that some counterfactual predictions are much more robust than others. Variation in even these simple cases suggests that RMAC can be a useful addition to the toolbox of structural estimation.

\subsection{Related Work}
Our work is closely related to the notion of partial identification \citep{manski2003partial}. The main idea behind partial identification is that many statistical models are only able to recover a set of parameters consistent with the data, not a single point estimate. The PI literature focuses on models where this `identified set' can be extracted easily. The adversarial revelation game is strongly related in that the equilibrium relaxation we employ makes the counterfactual predictions a set rather than a point. Our optimization procedure finds this set's worst (in terms of some evaluation function) and best elements and returns them.

Existing work in the field of market design has used econometric techniques to estimate counterfactuals in specific applications. Some existing work focuses on the econometrics of auctions and deriving underlying valuations (types) from bid behavior \citep{athey2010structural,chawla2017mechanism} or payment profiles. \cite{agarwal2015empirical} focuses on using data from medical residency matches to infer the underlying benefit to a young doctor from a particular residency. These approaches are, like ours, designed with the goal of answering counterfactual questions. However, while they allow for measures of statistical uncertainty they do not allow analysts to check for robustness of conclusions to violations of assumptions. \cite{haile2003inference} consider using `incomplete' models of auctions to provide some form of robustness but, like much of the literature on the econometrics of auctions (and unlike RMAC), requires hand-deriving estimators specifically tailored to the auction at hand.

Since the pioneering work of \cite{myerson1981optimal} there is a large subfield of game theory dedicated to designing mechanisms that optimize some quantity (e.g. seller revenue). Myerson-style results are useful because they give closed form solutions to optimal auction design, however, this comes at a high informational cost. For example, they often require the auctioneer to know the distribution of types (valuations) in the population. These strong assumptions are relaxed in robust mechanism design \citep{bergemann2005robust}, automated mechanism design \citep{conitzer2002complexity}, and recent work in using deep learning methods to approximate optimal mechanisms \citep{dutting2017optimal,feng2018deep}. Optimal mechanism design is related to, but different from, the RMAC problem as it typically assumes access to at least some direct information about the distribution of types, whereas our main problem is to robustly infer the underlying types from observed actions. However, these problems are related and combining insights from these literatures with RMAC is an interesting direction for future work.

There is recent interest in relaxing equilibrium assumptions in structural models. For example \cite{nekipelov2015econometrics} consider replacing equilibrium assumptions with the assumption that individuals are no-regret learners. This, again, gives a set valued solution concept which can be worked out explicitly for the special case of auctions. Given the prominence of no-regret learning in algorithmic game theory a natural extension of the work in this paper is to consider expanding RMAC to learning as a solution concept.

\section{Bayesian Games}
We consider the standard one-shot Bayesian game setup. There are $N$ players which each have a type $\theta_i \in \Theta$ drawn from an unknown distribution $\mathcal{F}$. This type is assumed to represent their preferences and private information. For example, in the case of auctions this type describes the valuations of each player for each object.

\begin{defn}
A game $\mathcal{G}$ has a set of actions for each player $\mathcal{A}_i$ with generic element $a_i$. After each player chooses their action, the players receive utilities given by $u^{\mathcal{G}}_i (a_1, \dots, a_N, \theta_i)$. 
\end{defn}

We will be interested in systems that come to a stable state and we will use the concept of Bayesian Nash equilibrium. We denote a strategy $\sigma_i$ for player $i$ in game $\mathcal{G}$ as a mapping which takes as input $\theta_i$ and outputs an action $a_i$. As standard for a vector $x$ of variables (strategies, actions, types, etc...), one for each player, we let $x_i$ be the variable for player $i$ and $x_{-i}$ be the vector for everyone other than $i$.

\begin{defn}
A Bayesian Nash equilibrium is a strategy profile $\sigma^{*}$ such that for each player $i$, all possible types $\theta_i$ for that player which have positive probability under $\mathcal{F}$, and any other strategy $\sigma'_i$ we have $$\mathbb{E}_{\mathcal{F}} \big[ u^{\mathcal{G}}_i (\sigma^{*}_i (\theta_i), \sigma^*_{-i} (\theta_{-i}), \theta_i) \big] \geq \mathbb{E}_{\mathcal{F}} \big[ u^{\mathcal{G}}_i (\sigma'_i (\theta_i), \sigma^*_{-i} (\theta_{-i}), \theta_i) \big].$$
\end{defn}

The Bayesian Nash equilibrium (BNE) assumption can be motivated by, for example, assuming that repeated play (with rematching) have led learning agents to converge to such a state \citep{fudenberg1998theory,dekel2004learning,hartline2015no}. Importantly, BNE states that players' actions are optimal given the distribution of partners they could play, not necessarily that they are optimal at each realization of the game with types fixed.

For the purposes of lightening notation from here on we will deal with games where every player's action set is the same $\mathcal{A}_i = \mathcal{A}$  and every players' type is drawn iid from $\mathcal{F}$.

\section{The Revelation Game as a Counterfactual Estimator}
Given the formal setup above, we now turn to answering our main question:

\begin{question}
Suppose we have a dataset $\mathcal{D}$ of actions played in $\mathcal{G}$. What can we say about what would happen if we changed the underlying game to $\mathcal{G}'$?
\end{question}

Formally, when we say that we change the game to $\mathcal{G}'$ we mean that the action set changes to $\mathcal{A}'$ and the utility functions change to $u^{\mathcal{G'}}_i (a_1, \dots, a_N, \theta_i).$ $\mathcal{G}'$ remains a Bayesian game so the definitions and notation above continue to apply. 

As a concrete example: in the case of online advertising auctions, $\mathcal{D}$ will contain a series of auctions with bids taken by different participants. We may wish to ask, what would happen if we changed the auction format?

We now discuss a set of assumptions typically made either implicitly or explicitly when analysts apply equilibrium based structural models to estimate a counterfactual:

\begin{ass}[Equilibrium]
Data is drawn from a BNE of $\mathcal{G}$ and play in $\mathcal{G}'$ will form a BNE.
\end{ass}

\begin{ass}[Identification]
For any possible distribution of types $\mathcal{F}$ and associated BNE $\sigma^*$ there does not exist another distribution of types $\mathcal{F}'$ and BNE $\sigma'^*$ that induces the same distribution of actions.
\end{ass}

\begin{ass}[Uniqueness in $\mathcal{G}'$]
Given $\mathcal{F}$ there is a unique BNE in $\mathcal{G}'.$
\end{ass}

If the assumptions are satisfied then we can use $\mathcal{D}$ to answer the counterfactual question. By Assumption 1 each action $d_i$ is optimal against the distribution of actions implied by $\mathcal{D}.$ If $\mathcal{D}$ is large enough then it approximates the true distribution implied by $\sigma$ and $\mathcal{F}$. By Assumption 2, there is a unique $\sigma$ and $\mathcal{F}$ that give rise to this distribution and we can use various methods to find them. Once we have $\mathcal{F}$ we can solve for the equilibrium in $\mathcal{G}'$, which is unique by Assumption 3, using any number of methods and we are done.

We now show this procedure is equivalent to solving for the Nash equilibrium in a modified game which we refer to as a \textit{revelation game}.\footnote{We are indebted to Jason Hartline who pointed out in an earlier versions of this work that our optimization problem can be thought of as equilibrium finding and thus make exposition much simpler.} We do not consider that agents will actually play this game, rather we will show that this proxy game is a useful abstraction for doing robust counterfactual inference.

The revelation game has $m$ players, one for each element of $\mathcal{D}$. We refer to these as data-players to avoid confusion with the players in $\mathcal{G}$ and $\mathcal{G}'$. Each data-player knows that the analyst has a random variable $\mathcal{D}$ of actions from the equilibrium of $\mathcal{G}$. $\mathcal{D}$ includes the data-player's own true equilibrium action but the other actions are ex-ante unknown. Each data-player has a true type $\theta_j$ which is unknown to the analyst, the types of the other data-players $-j$ are unknown to $j$ but it is commonly known that they are drawn from the distribution $\mathcal{F}.$

Each data-player $j$ makes a decision: they report a type $\hat{\theta}_j$ and an action for the counterfactual game $\hat{a}_j.$ They are paid as follows: first, let the $\mathcal{D}_{-j}$ denote the random variable which denotes the actions of the other data-players the analyst will observe. Now we define the $\mathcal{G}$-Regret of data-player $j$ as $$\text{Regret}^{\mathcal{G}}_{j} (\hat{\theta}_j, \mathcal{D}_{-j}) =  \text{max}_{a_j} \mathbb{E} \big[ u^{\mathcal{G}}_j (a_j, \hat{\theta}_j, \mathcal{D}_{-j}) \big] - \mathbb{E} \big[ u^{\mathcal{G}}_j (d_j, \hat{\theta}_j, \mathcal{D}_{-j}) \big].$$

We define the $\mathcal{G}'-$Regret of data-player $j$ as $$\text{Regret}^{\mathcal{G}'}_{j} (\hat{a}_j, \hat{\theta}_j, \hat{a}_{-j}) = \text{max}_{a_j}  \big[ u^{\mathcal{G}'}_j (a_j, \hat{\theta}_j, {\hat{a}_{-j}}) \big] - \big[ u^{\mathcal{G}'}_j (\hat{a}_j, \hat{\theta}_j, \hat{a}_{-j}) \big].$$

The revelation game is a Bayesian game where each data-player $j$ tries to minimize a loss given by the max of the two above regrets: $$\mathcal{L}^{rev}_j (\hat{\theta}_j, \hat{a}_j, \hat{a}_{-j}, \mathcal{D}) = \max \lbrace \text{Regret}^{\mathcal{G}}_{j} (d_j, \hat{\theta}_j, \mathcal{D}), \text{Regret}^{\mathcal{G}'}_{j} (\hat{a}_j, \hat{\theta}_j, \hat{a}_{-j}) \rbrace.$$

Since most game theory definitions (e.g. equilibria) use utility maximization, rather than loss minimization, we will also sometimes use the notation $$\mathcal{U}^{rev}_j (\hat{\theta}_j, \hat{a}_j, \hat{a}_{-j}, \mathcal{D}) = -\mathcal{L}^{rev}_j (\hat{\theta}_j, \hat{a}_j, \hat{a}_{-j}, \mathcal{D}).$$

Given these definitions, we can show the following property: 

\begin{thm}\label{revgame}
If assumptions 1-3 are satisfied then the revelation game has a unique BNE where each agent reveals their true type and counterfactual action.
\end{thm}

We leave the proof of the theorem to the Appendix. With this result in hand, we now discuss how to modify the revelation game to make our counterfactual predictions robust.

\section{Robust Multi-agent Counterfactual Inference}
In reality, assumptions 1-3 above are rarely satisfied exactly and we would like to see how robust our conclusions are to violations of our assumptions. In particular, we are interested in allowing agents to not be perfectly rational, not requiring identification to hold strictly, and allowing for our model of agents' reward functions to be misspecified. In addition, all modeling makes the important assumption

\begin{ass}[Specification]
$\mathcal{G}$ and $\mathcal{G}'$ include the correct specifications of individuals' reward functions.
\end{ass}

which, like the others, is rarely completely true in practice.

To relax all of these assumptions we will consider the concept of $\epsilon$-equilibrium:

\begin{defn}
For $\epsilon > 0$ an $\epsilon$-Bayesian Nash equilibrium is a strategy $\sigma^{*}$ such that for each player $i$, all possible types $\theta_i$ for that player which have positive probability under $\mathcal{F}$, and any other strategy $\sigma'_i$ we have $$\mathbb{E}_{\mathcal{F}} \big[ u_i (\sigma^{*}_i (\theta_i), \sigma^*_{-i} (\theta_{-i}), \theta_i) \big] \geq \mathbb{E}_{\mathcal{F}} \big[ u_i (\sigma'_i (\theta_i), \sigma^*_{-i} (\theta_{-i}), \theta_i) + \epsilon \big].$$
\end{defn}
	
Allowing for $\epsilon$-BNE in the revelation game means that we are also allowing for $\epsilon$-BNE in $\mathcal{G}$ and $\mathcal{G}'$. Using this formulation is equivalent to relaxing assumptions about rationality or correct specification in our structural models. $\epsilon$-equilibria can arise because agents are imperfect optimizers (but are able to learn to avoid actions that cause huge negative regret) or because the utility functions in $\mathcal{G}$ or $\mathcal{G}'$ are slightly incorrect (and individuals reach an equilibrium corresponding to some other reward function).

However, like many instances of partial identification \cite{manski2003partial} $\epsilon$-BNE is a set valued solution concept. Rather than enumerate the whole set, we will consider particular boundary equilibria:

We assume the existence of an evaluation function $V(\theta, a)$  which gives us a scalar evaluation of the counterfactual outcome that the analyst cares about. We overload notation and let $V(\sigma)=\mathbb{E}_{(\theta, a) \sim \sigma}V(\theta, a)$ be the expected value of $V$ given a mixed strategy $\sigma$. Common examples of valuation function used in the mechanism design literature include revenue, efficiency, fairness, envy, stability, strategy-proofness, or some combination of them \citep{roth1992two,guruswami2005profit,budish2011combinatorial,caragiannis2016unreasonable}. 

We will consider the maximal and minimal elements of the $\epsilon$-BNE set with respect to $V$. Formally:

\begin{defn}
  The $\epsilon$-pessimistic counterfactual prediction of $V$ is
  \[
    \inf_{\sigma} V(\sigma) \text{ s.t. }\sigma \text{ is an $\epsilon$-BNE in the revelation game.}
  \]
  The $\epsilon$-optimistic prediction replaces the inf with sup. The $\epsilon$-RMAC bounds are the values of $V$ attained at the pessimistic and optimistic predictions.
\end{defn}
 
 Figure \ref{rmac_summary} summarizes the idea behind RMAC. Standard structural imply a one-to-one mapping between observed distributions and underlying types followed by a one-to-one mapping between underlying types and counterfactual behavior. Assuming only $\epsilon$-equilibrium makes both of these mappings one-to-many and RMAC bounds select the most optimistic and pessimistic counterfactual distributions consistent with these mappings.
 
 \begin{figure*}[ht!]
  \centering
   \includegraphics[scale=.5]{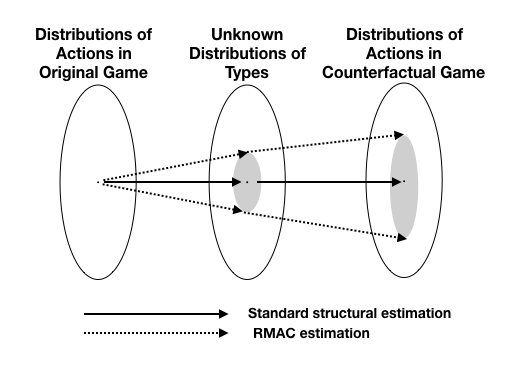}
   \caption{Standard structural assumptions allow us to map an observed distribution of equilibrium actions into an underlying distribution of types and then map this distribution of types into a distribution of counterfactual actions. This is represented by points and solid arrows. Assuming $\epsilon$-equilibrium in the original game means we now map a distribution of actions to many possible distributions of types. Assuming $\epsilon$-equilibrium in the counterfactual game implies another one-to-many mapping from underlying type distributions to counterfactual action distributions. RMAC bounds, represented by the dashed arrows, are the most optimistic and pessimistic choices of counterfactual actions (with respect to some evaluation function $V$) from this set.}
  \label{rmac_summary}
\end{figure*}

\section{Computing Equilibria of the Revelation Game}\label{MIPsection}
In practice, we can replace the random variable $\mathcal{D}$ of the revelation game with their sample analogue, the observed data. From here forward $\mathcal{D}$ will refer to the sample data. Unfortunately, we can derive a quite negative complexity result for computing $\epsilon$-RMAC bounds exactly:

\begin{thm}\label{rfpishard}
It is NP-hard to compute the robust counterfactual estimate even if each data-point $j$ has only a single feasible type, and there are only two data points. It is also NP-hard even if there is no objective function,  a finite number of feasible types, and $\mathcal{G}'$ has only two players.
\end{thm}

The proof follows from the reduction of solving the revelation game to other known NP-hard problems and we leave it to the Appendix. Importantly, NP-hard does not mean impossible and the RMAC bounds can be computed for small instances using a mathematical program. We give this program in the Appendix for the case where we only consider pure-strategy $\epsilon$-BNE. In addition, we show that for the special case where $\mathcal{G}'$ is a two player game, we can solve for the RMAC bounds using a mixed integer program.

\subsection{Revelation Game Fictitious Play}
Given that computing RMAC programs is intractable beyond the simplest cases, we propose to adapt the fictitious play algorithm \cite{brown1951iterative} to compute the optimistic and pessimistic equilibria of the revelation game. We refer to this as Revelation Game Fictitious Play (RFP). 

RFP works as follows. For notation, let $\hat{\theta}^t_{i}$ be the estimated type for data point $i$ at iteration $t$ and $\hat{a}^t_{i}$ be the estimated counterfactual action at iteration $t.$ Recall the definition of the revelation game utility as $$\mathcal{U}^{rev}_i (\hat{\theta}_i, \hat{a}_i, \hat{a}_{-i}, \mathcal{D}) = -\max \lbrace \text{Regret}^{\mathcal{G}}_{i} (d_i, \hat{\theta}_i, \mathcal{D}), \text{Regret}^{\mathcal{G}'}_{i} (\hat{a}_i, \hat{\theta}_i, \hat{a}_{-i}) \rbrace.$$ 

As with standard fictitious play, at each time step each $i$ reports a type-action pair. They observe the choices of others and update their $t+1$ choice to be $(\hat{\theta}^{t+1}_i, \hat{a}^{t+1}_i)$ to be the one that minimizes (or maximizes) $V$ out of the set of $\epsilon$ best responses to the current history of play (when $\epsilon=0$ RFP simply chooses the best response to the current history, breaking ties randomly). The pseudocode is shown in algorithm 1.

\begin{algorithm}\label{algRFP}
\caption{Revelation Fictitious Play}
\begin{algorithmic}
\State \textbf{Input:} $\epsilon, \mathcal{D}, V, \mathcal{G}, \mathcal{G}'$, pessimism/optimism
\If{pessimistic}
\State $\alpha \gets -1$
\EndIf
\If{optimisic}
\State $\alpha \gets 1$
\EndIf
\State Randomly initialize $\hat{\theta}^0_{i}, \hat{a}^0_{i}$
\While {not converged}
\State Let $\bar{a}^t_{-i}$ be the historical distribution of $\hat{a}^t_{-i}$ for $t \in \lbrace 0, \dots, t \rbrace$
\State Let $\bar{\theta}^t_{-i}$ be the historical distribution of $\hat{\theta}^t_{-i}$ for $t \in \lbrace 0, \dots, t \rbrace$
\State Let $\bar{\sigma}^t_{-i}$ be the strategy defined by the historical distribution of $(\hat{\theta}^t_{-i}, \hat{a}^t_{-i})$ for $t \in \lbrace 0, \dots, t \rbrace$
\State Let the set of low-regret revelation game actions be $$\hat{\mathcal{C}}_i^t = \lbrace (\hat{\theta}, \hat{a}) \in \mathcal{A} \times \Omega \mid \mathcal{L}^{rev}_i (\hat{\theta}^t_{i}, \hat{a}^t_{i}, \bar{a}^t_{-i}, \mathcal{D}) \leq \epsilon \rbrace$$
\State Breaking ties randomly, update guesses for each datapoint
\[
  (\hat{\theta}_{i}^{t+1}, \hat{a}_{i}^{t+1}) = \text{argmax}_{\hat{a}, \hat{\theta} \in \hat{\mathcal{C}}^t_{i}} \left[ \alpha V(\hat{a}_i, \hat{\theta}_i, \bar{\sigma}^t_{-i}) \right].
\]
\EndWhile
\end{algorithmic}
\end{algorithm}

It is well-known that fictitious play converges in $2$-player zero-sum and potential games, while it may cycle in general. Nonetheless, a well-known result states that \emph{if} fictitious play converges, then it converges to a Nash equilibrium \citep{fudenberg1998theory}. 

We now show an analogous result for RFP: if RFP converges then it converges to an $\epsilon$-BNE and locally minimizes $V$ in the sense that no unilateral deviation by a single data-player $j$ in the revelation game that are \emph{strictly} $\epsilon$-best responses leads to a smaller $V$. 

Recall that we use the notation $(\bar{\theta}^t, \bar{a}^t)$ to denote a \textit{history} of behavior. We denote by $\bar{\sigma}^t$ the \textit{mixed strategy}  implied by that history. As with standard fictitious play we consider convergence of $\bar{\sigma}^t$:

\begin{defn}
  RFP \emph{converges} to a mixed strategy $\sigma^*$ if
  $$
    \lim_{t \rightarrow \infty} \bar{\sigma}^t  = \sigma^* 
  $$
\end{defn}

% We will abuse notation slightly and let $V(\theta^*, a^*)$ be the expected value of $V$ for a distribution $(\theta^*, a^*)$.

We use the following notion of local optimality (analogously defined for optimistic V):
\begin{defn}
  A mixed $\epsilon$-BNE $\sigma^*$ of the revelation game is \emph{locally $V$-optimal} if
  \[
    V(\sigma^*) \leq V(\theta_j, a_j, {\sigma}_{-j}^*)
  \]
  for any data-player $j$ and unilateral deviation $(\theta_j, a_j)$ where $$\mathbb{E}_{(\theta_{-j}, a_{-j}) \sim \sigma_{-j}^*}[\mathcal{U}_j^{rev}(\theta_j, a_j, a_{-j}, \mathcal{D})] < \epsilon.$$
\end{defn}

Note the strict inequality on the value of the deviation: we do not show robustness to lower $V$ at deviations that are on the boundary of the $\epsilon$-best-response set at convergence. The reason is that there may be deviations which have strictly greater than $\epsilon$ regret for all $t$, but their regret converges to $\epsilon$ from above, and so they enter the set at the limit.

\begin{thm}\label{rfpconverge}
  If RFP converges to $\sigma^*$ then $\sigma^*$ is an $\epsilon$-BNE of the revelation game, and locally optimal.
\end{thm}

We relegate the proof to the Appendix. The argument is a fairly straightforward extension of standard fictitious play results to the revelation game. 

An important question is whether RFP can be guaranteed to converge in particular classes of Bayesian games. We leave the theoretical study of RFP (or other learning algorithms in the revelation game) to future work and focus the rest of the paper on empirical evaluation. 

\section{Experiments}
We now turn to constructing RMAC bounds for classic problems in market design including auctions, school choice, and social choice.

\subsection{RMAC in Auctions}
\begin{figure*}[ht!]
  \centering
   \includegraphics[scale=.5]{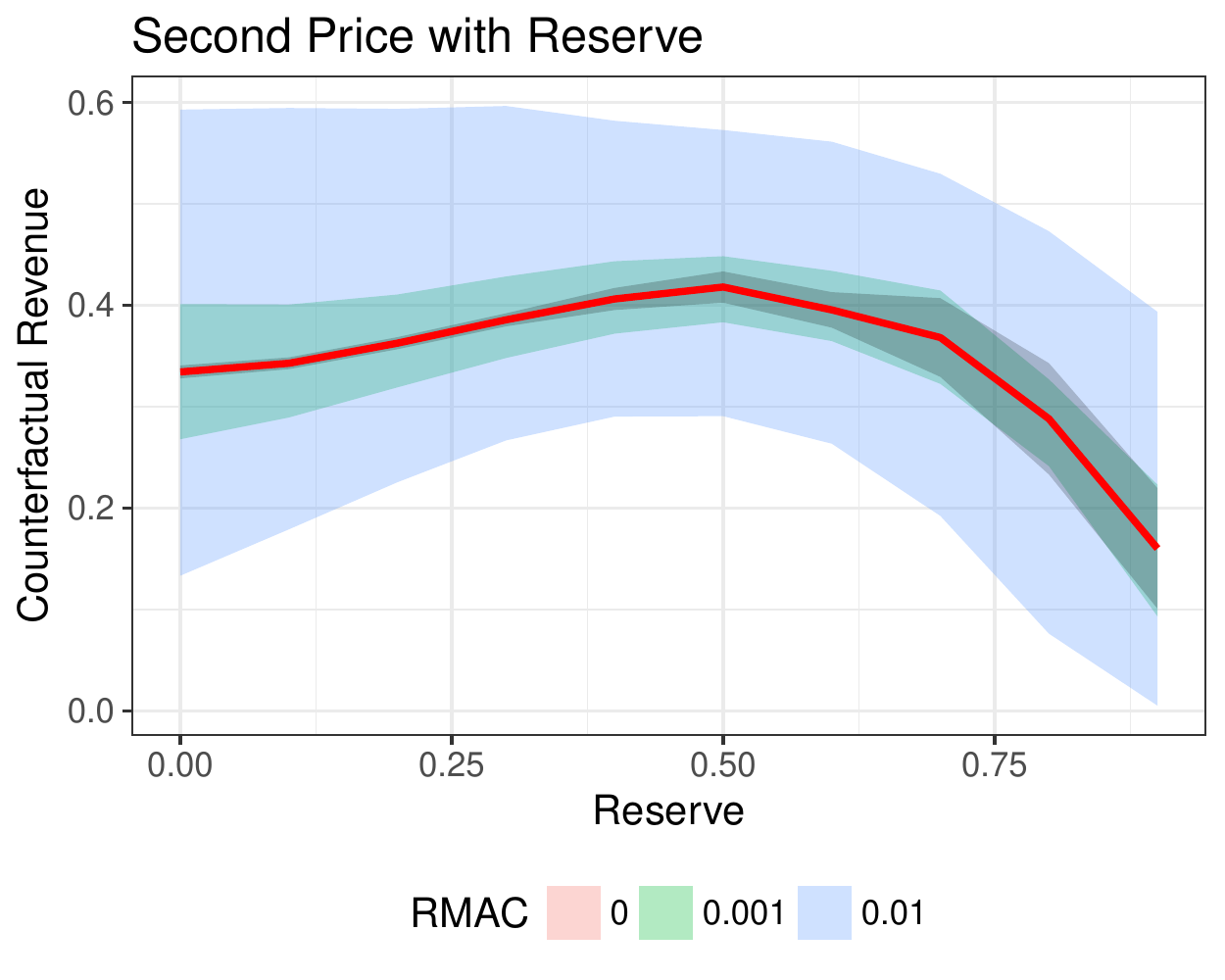}
      \includegraphics[scale=.5]{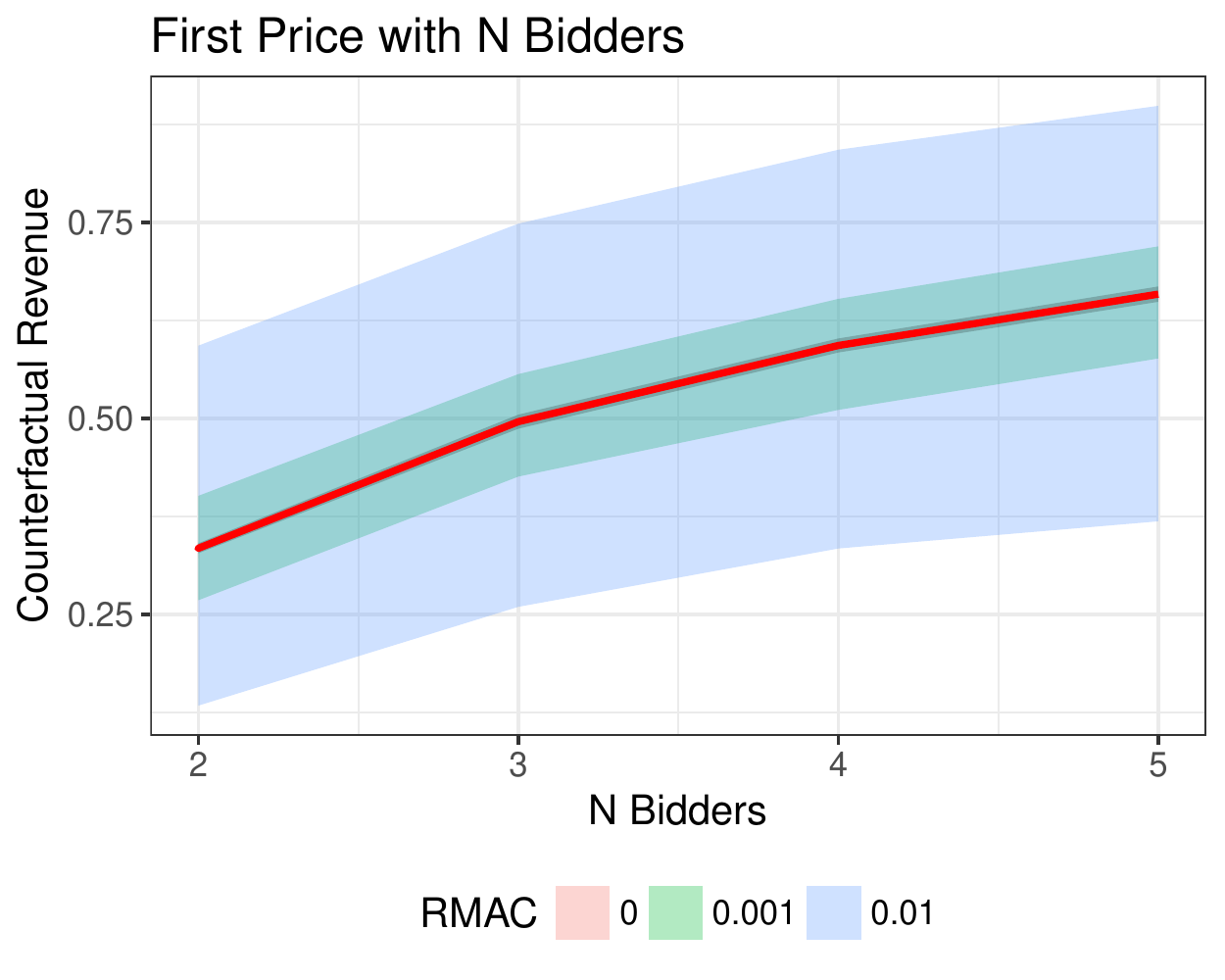} 
   \caption{RMAC revenue predictions using data drawn from the equilibrium of a first price $2$ player auction for various counterfactual auction formats. The RMAC robustness bounds, even with small $\epsilon$ are much larger than the standard error bounds (grey ribbon around RMAC $0$ line) estimated from multiple replicates.}
  \label{auctionexpt}
\end{figure*}

As our first evaluation, we will consider the study of counterfactual revenue in auctions. We set $\mathcal{G}$ as a first-price $2$-player auction with types drawn from $[0,1]$ uniformly and bids in the interval $[0,1]$ discretized at intervals of $.01$. As our counterfactual games we consider a $2$-player second-price auction with varying reserves\footnote{A reserve price $r$ in an auction is a price floor, individuals cannot win the auction if they bid below the reserve. In addition, in the case of second-price auctions, the price paid by the winner is the max of $r$ (as long as $r$ is less than the bid) and the second-highest bid.} in the interval $[0,1]$ and $N$ player first-price auctions. 

We use counterfactual expected revenue as our valuation function. We set the domain of possible types to also be equal to $[0,1].$\footnote{In our experiments we found that the choice of initial hypothesis space mattered very much, allowing a larger upper bound let some extreme types to be set fairly far above $1$. Thus, incorporating analyst priors is an important part of RMAC and the addition of other forms of regularization into the procedure that can reflect these priors is an important future research direction.}

We generate data by first sampling 1000 independent types and their actions from the closed form first-price equilibrium ($bid = .5 \theta$), using these actions as $\mathcal{D}$. We then use $\mathcal{D}$ to compute $\epsilon$-RMAC predictions for several levels of $\epsilon$. Figure \ref{auctionexpt} shows our results with (small) error bars being shown as standard deviations of the statistic over replicates. 

\begin{figure*}
  \centering
\includegraphics[scale=.5]{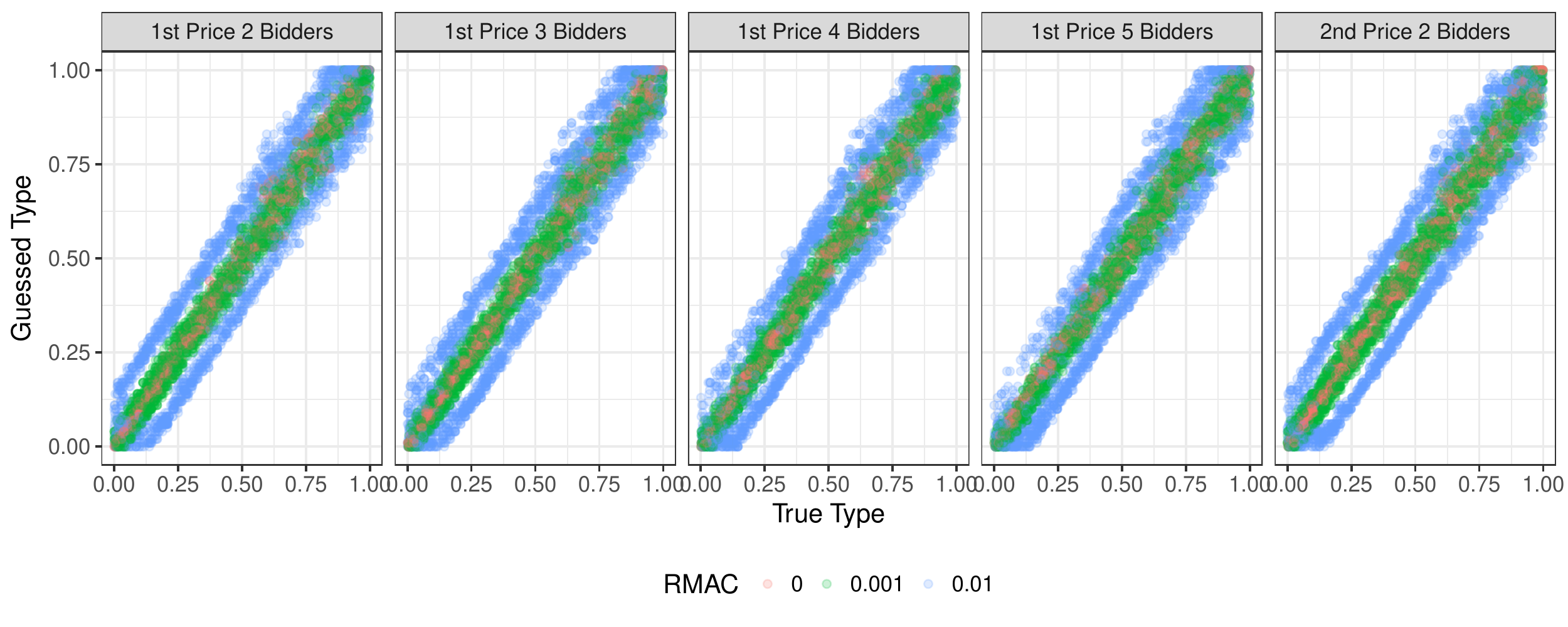} 
   \includegraphics[scale=.5]{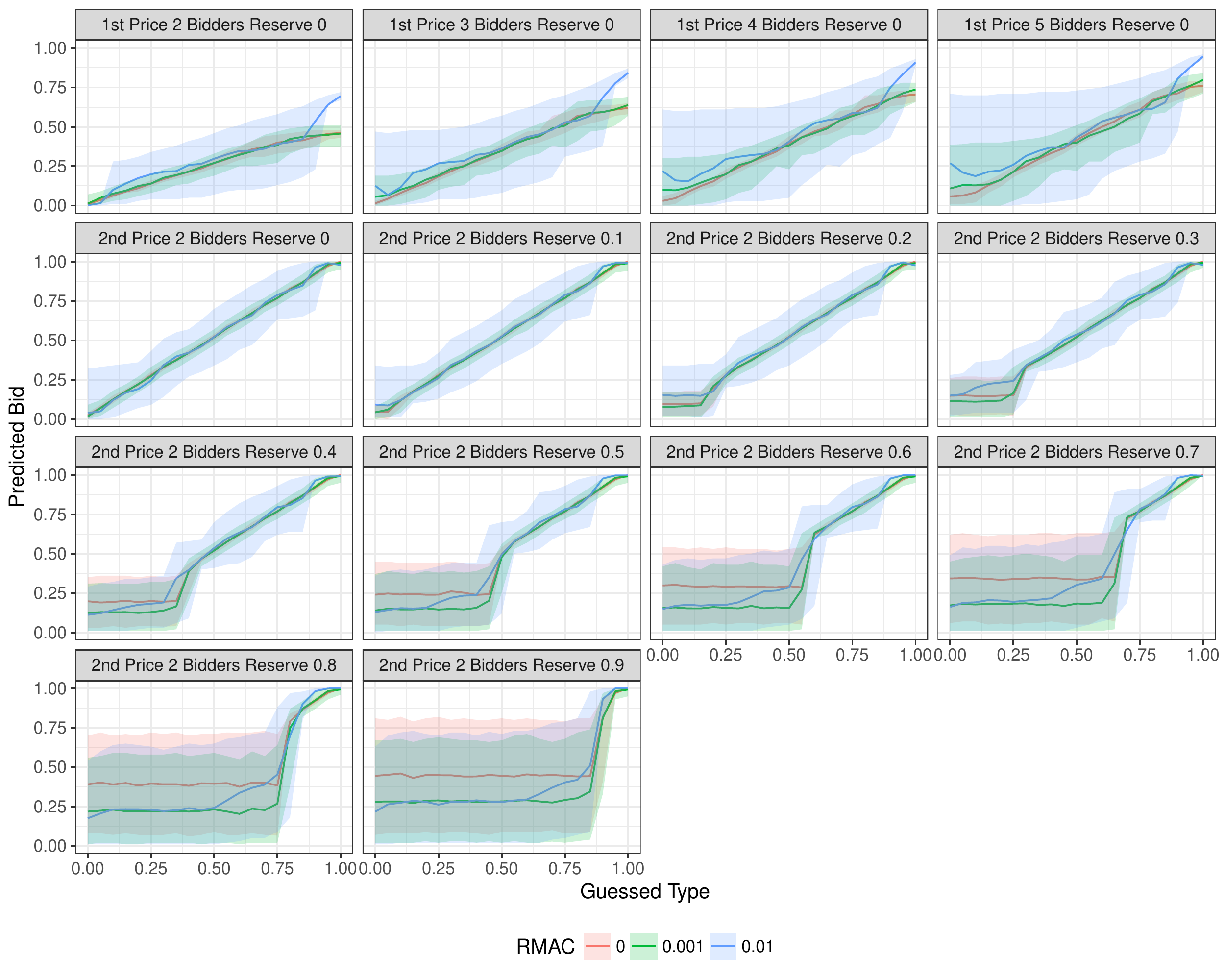}

   \caption{In depth analysis of how RMAC changes counterfactual estimates. Top panel shows estimated types and RMAC pushes the entire distribution up or down, in this special case the extent of the downward shift is not affected by the counterfactual game. This happens in auctions because in RMAC the type regret is determined by $\mathcal{G}$ and $\mathcal{D}$ and lower valuations will guarantee lower $V$ in the counterfactual game. Bottom panel shows RMAC generated counterfactual strategies for various counterfactual auctions.}
  \label{auctiontypes}
\end{figure*}

We see that in auctions, even slight changes to $\epsilon$ can lead to larger changes in revenue. In particular, if we consider that the average expected utility accrued to the winner in the $2$ player auction is $.25$, an $\epsilon$ of $.01$ corresponds only to a $4 \%$ misoptimization/misspecification. However, this small $\epsilon$ still gives quite wide revenue bounds. 

To see the logic behind this lack of robustness, consider the pessimistic estimate, in which the data is drawn from an $\epsilon$-equilibrium where individuals are overbidding in the original game and underbidding in the counterfactual game.

Assuming a uniform bid distribution, an individual's regret for (unilaterally) deviating by $\Delta$ is $\epsilon=\Delta^2/2$ in either a first- or second-price auction. However, if all individuals decrease their bid by $\Delta$ expected revenue will decrease by $\Delta$. Therefore, we expect a worst-case $\epsilon$-equilibrium in the counterfactual game to decrease revenue by $\sqrt{2\epsilon}$. In addition, there will be a similar decrease in revenue from the shift in types inferred from the original game.

The top panel of Figure \ref{auctiontypes} plots the RMAC estimated types as a function of true type and we can see that the type distribution is fairly uniformly shifted down. As a robustness check we can also see that this downward shift is not affected by the counterfactual game. This is not a general property of the RMAC estimator, and is specific to this case of auctions where revenue will be monotonic in counterfactual bid and counterfactual bid will be monotonic in type.

The worst case scenario is compounded by assumption that the equilibrium that attains in the counterfactual will be the one where these same individuals will slightly underbid. We can see the RMAC type-contingent counterfactual strategies plotted in the bottom panel of figure \ref{auctiontypes}. Error ribbons reflect $10^{th}$ and $90^{th}$ percentiles taken over multiple replicates with wide bands appearing when reserves are set high since any bid below the reserve always achieves a payoff of $0$ and so individuals are indifferent between those bids.

\subsubsection{RMAC Without Point Identification}
We now discuss how RMAC can be useful for situations where point identification of a structural model is not guaranteed. This can happen when there are multiple equilibria in $\mathcal{G}'$ or when the mapping from type distributions to equilibrium distributions in $\mathcal{G}$ is not injective. In such situations there will be multiple solutions to a maximum likelihood estimator and no guarantees about which one will be output by the procedure. On the other hand, RMAC bounds will still be well defined and if we choose a small enough $\epsilon$ will be close to the worst and best case full equilibria.

We illustrate this by considering counterfactual prediction where $\mathcal{G}$ is a $2$ player second-price auction with reserve $.5.$ with the same simulation parameters as above (as $\mathcal{D}$ we use truthful reports). $\mathcal{G}$ is dominant strategy truthful for all types $\theta > .5$ but the payoff to bids in the interval $[0, .5]$ is always $0$ so any type $\theta < .5$ can rationalize any bid in this interval. This means that the type distribution is not point identified from an action distribution. We apply RMAC to this situation with the counterfactual question of what would happen if we changed the reserve $r$.

Figure \ref{auctionnotid} shows the results. We see on the left panel that RMAC bounds for reserves $[0, .5]$ are very wide whereas bounds for reserves above the original $.5$ are smaller since our type censoring appears only on one side. 

The right panel shows that here, unlike in the auction experiments above, the choice of $\mathcal{G}'$ does affect type estimation. When the counterfactual reserve is $0$ then the pessimistic RMAC pushes previously unidentified types to $0$ to create the worst case $\mathcal{G}'$ equilibrium. When the counterfactual reserve is set very high to $.9$ low types do not bid above the reserve even in the optimistic $\epsilon$ equilibria and so types which were not identified in the original $r=.5$ game remain unidentified and their guesses are chosen arbitrarily.

\begin{figure*}[ht!]
  \centering
   \includegraphics[scale=.5]{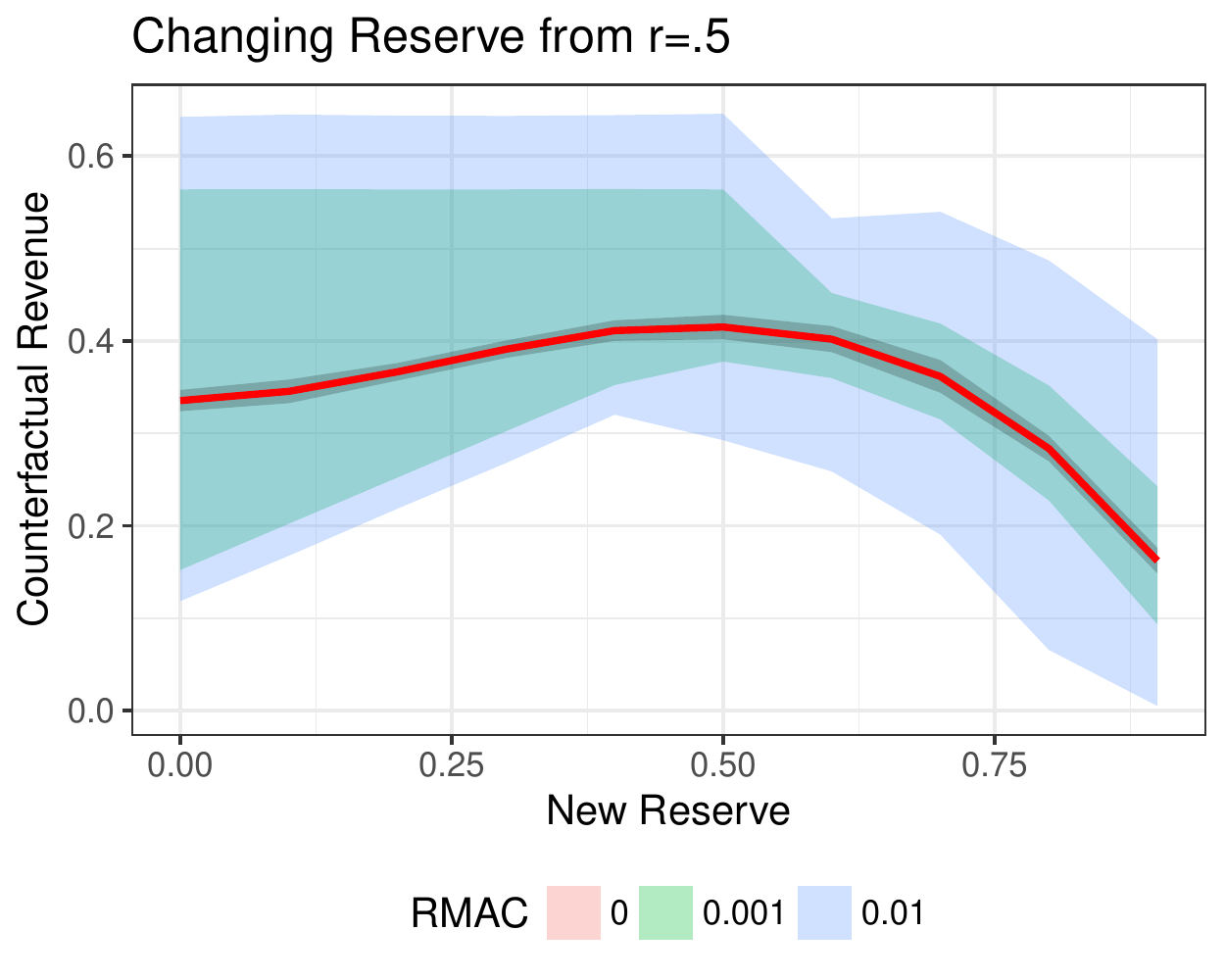}
      \includegraphics[scale=.5]{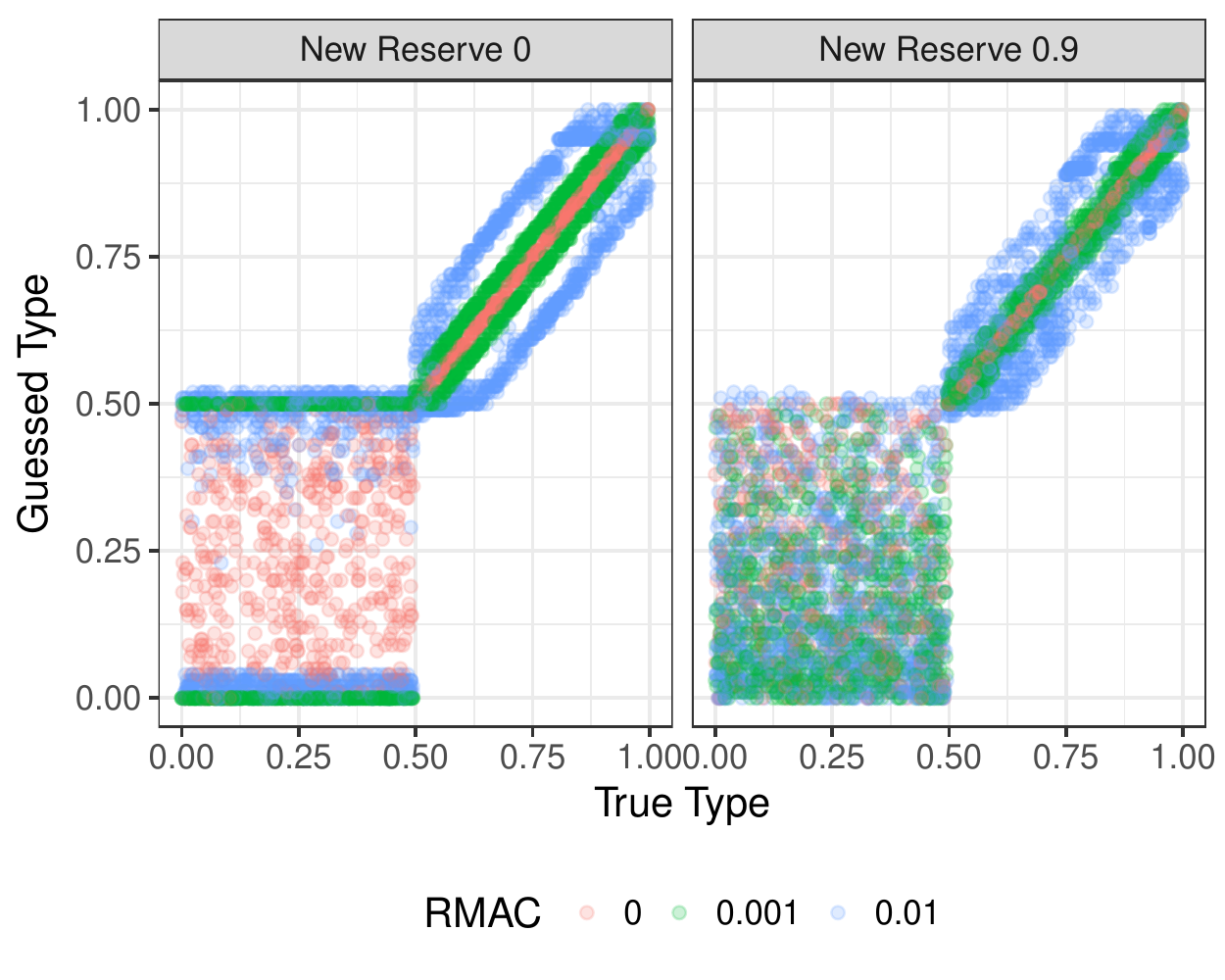} 
   \caption{Results for data drawn from a second-price auction with reserve $.5$ with counterfactual question involving changing the reserve. RMAC is well defined even when the inverse problem is not identified due to multiple types being consistent with the same observed actions. The maximum likelihood solution (red line) simply picks a random type from among all equally likely ones. RMAC bounds reflect the lack of identification in the original game as they are quite large for counterfactual reserves less than the original reserve. In the right panel we see that in this situation, unlike in the example above, the choice of counterfactual game $\mathcal{G}'$ does affect the estimated underlying types.}
  \label{auctionnotid}
\end{figure*}

\subsection{RMAC in School Choice}
We move to another commonly studied domain: school choice. Here the problem is to assign items (schools) to agents (students). Agents have preferences over schools, report them, and the output of the mechanism is an assignment.

We look at two real world school choice mechanisms. The first is the Boston mechanism \citep{abdulkadirouglu2005boston}. In Boston each student reports their rank order list and the mechanism tries to maximize the number of first choice assignments that it can. Once it has done this, it tries to maximize the number of second-choice assignments, and so on. The second mechanism uses the random serial dictatorship (RSD) mechanism \citep{abdulkadirouglu1998random}. Here students are each given a random number and sorted, the first in line gets to choose their favorite school, the second chooses their favorite among what's left and so on. 

\begin{figure*}[ht!]
  \centering
   \includegraphics[scale=.5]{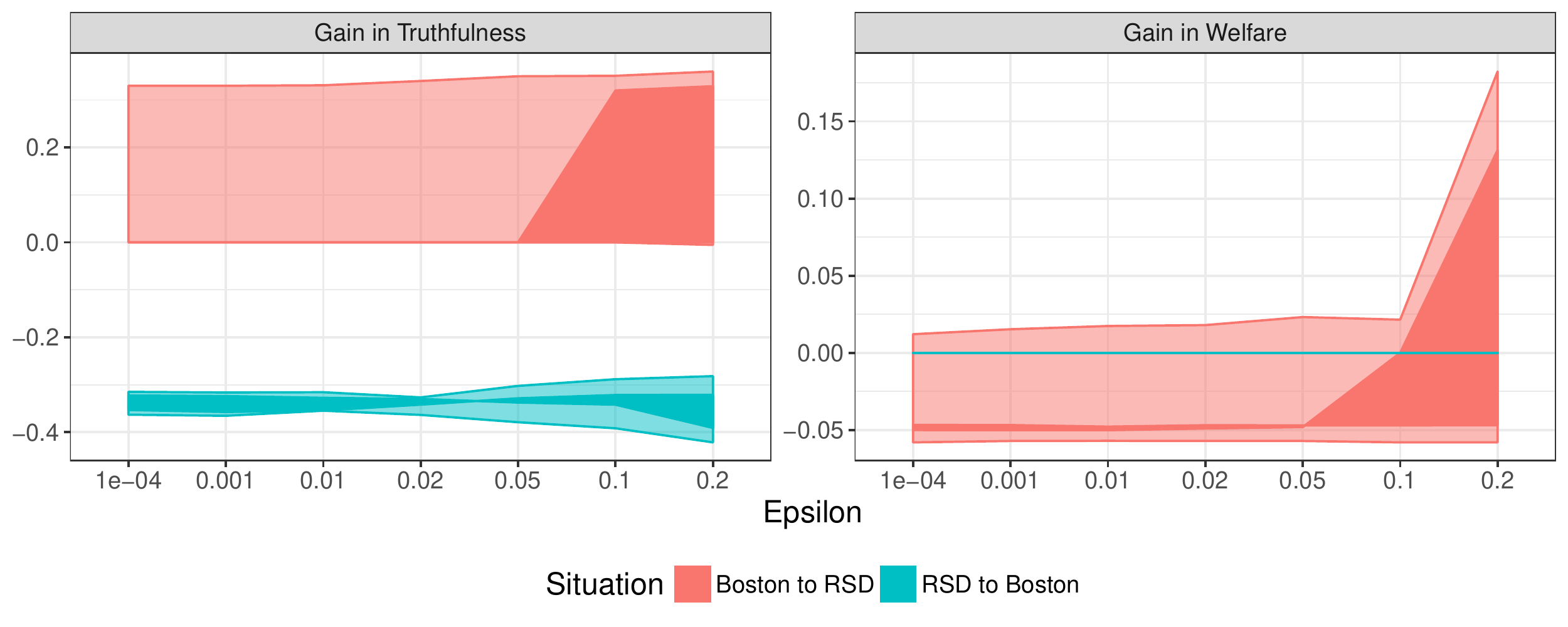}
   \caption{RMAC intervals for the change in social welfare and change in truthfulness from changing school choice mechanisms. Dark and light curves are for $10^{\rm th}$ and $90^{\rm th}$ percentile of estimated intervals over replicates with different sampled $\mathcal{D}$. The presence of multiple type distributions consistent with a given action distribution in Boston means that even for small $\epsilon$ RMAC bounds can be quite wide for Boston to RSD.}
  \label{schoolexpt}
\end{figure*}

The main tradeoff in practical school choice comes from balancing the total social welfare achieved by the mechanism and the strategy proofness. RSD (and other algorithms like student-proposing deferred acceptance) have a dominant strategy for each agent to report their true type. This means that participants in real world implementations of such mechanisms do not need to spend cognitive effort on guessing what others might do or searching for information - they can simply tell the truth and go on with their day. On the other hand, equilibria of the Boston mechanism can be more efficient in terms of allocating schools to students but are not strategy proof \citep{mennle2015trade,abdulkadirouglu2011resolving}. 

We consider a problem mechanism with 3 students and 3 schools $(A, B, C)$. For both mechanisms the action space is a permutation over $A, B, C$. 

We consider a hypothesis space of types that are permutations of utility vector $(5, 4, 0)$ - that is, individuals receive utility $5$ if they get their first choice, $4$ for the second and $0$ for the third. We are going to consider the case where all individuals have identical preferences of $A > B > C.$ We will take these types, generate an equilibrium under Boston and construct a dataset and ask the counterfactual question: what would happen if we changed to RSD? We will also generate actions from the RSD equilibrium and ask: what would happen if we changed to Boston?

We look at two choices of $V$ inspired by discussion of these mechanisms: overall social welfare of the allocation and truthfulness of the strategies (i.e. whether types report their true values). We plot the estimated change in welfare and truthfulness from moving from one mechanism to another. In other words, we perform the kind of exercise that a practicing market designer might actually do to when trying to convince a school district to change mechanisms. 

Note that in the case of `Boston to RSD' at $\epsilon=0$ the standard structural assumptions are not satisfied, as multiple type distributions are consistent with the observed actions. Given our utility space, even though everyone has the same preferences, same types may choose different actions (i.e. play a mixed strategy), since it is better to be assured of getting $B$ than take a lottery between $A$, $B$ and $C.$ So, some proportion of individuals will report $(B, A, C)$ However, such an action profile is also consistent with an equilibrium of truthful types with different preferences. Since the types are not identified from the observed actions, structural estimation using maximum likelihood has multiple optima with different values of $V$. However, RMAC with small $\epsilon$ will produce an interval that covers both possible type distributions.

Figure \ref{schoolexpt} shows that going from Boston to RSD can create more truthfulness in the best case but in the pessimistic case has no effect (because actions were already truthful). This transition also tends to lead to welfare decreases, although not always. For example, if all players have identical preferences, all mechanisms provide the same welfare. Moving from RSD to Boston decreases truthfulness always and does not change welfare in our situation (since in our simulations all students have same true preferences).  

%At small $\epsilon$, Figure \ref{schoolexpt} shows a range of RMAC estimates across replicates. We assumed that participants are playing a mixed equilibrium, which requires that they are indifferent between actions. But due to the finite size of $\mathcal{D}$, participants are not exactly indifferent, which means that some type distributions will not be in $\epsilon$-equilibrium for small $\epsilon$. As $\epsilon$ increases, RMAC predicts a similar range of $V$ across replicates.

\subsection{RMAC in Social Choice}
As our last study we move to the domain of social choice. We consider the standard example of a group of individuals choosing an ideal point $x^* \in [0,1].$ We assume individuals have a type $\theta \in [0,1]$ and have single-peaked preferences and receive loss $(x^* - \theta)^2$ from a point $x^*$ being chosen for the group. We consider groups of 11 individuals participating in one of $3$ mechanisms: in each mechanism individuals report a number $a \in [0,1]$. In the \textit{mean} mechanism, $x^*$ is chosen as the mean of the reports, in the \textit{median} mechanism the median is chosen. In both the median and mean mechanism no side payments are made. 

\begin{figure*}[ht!]
  \centering
   \includegraphics[scale=.5]{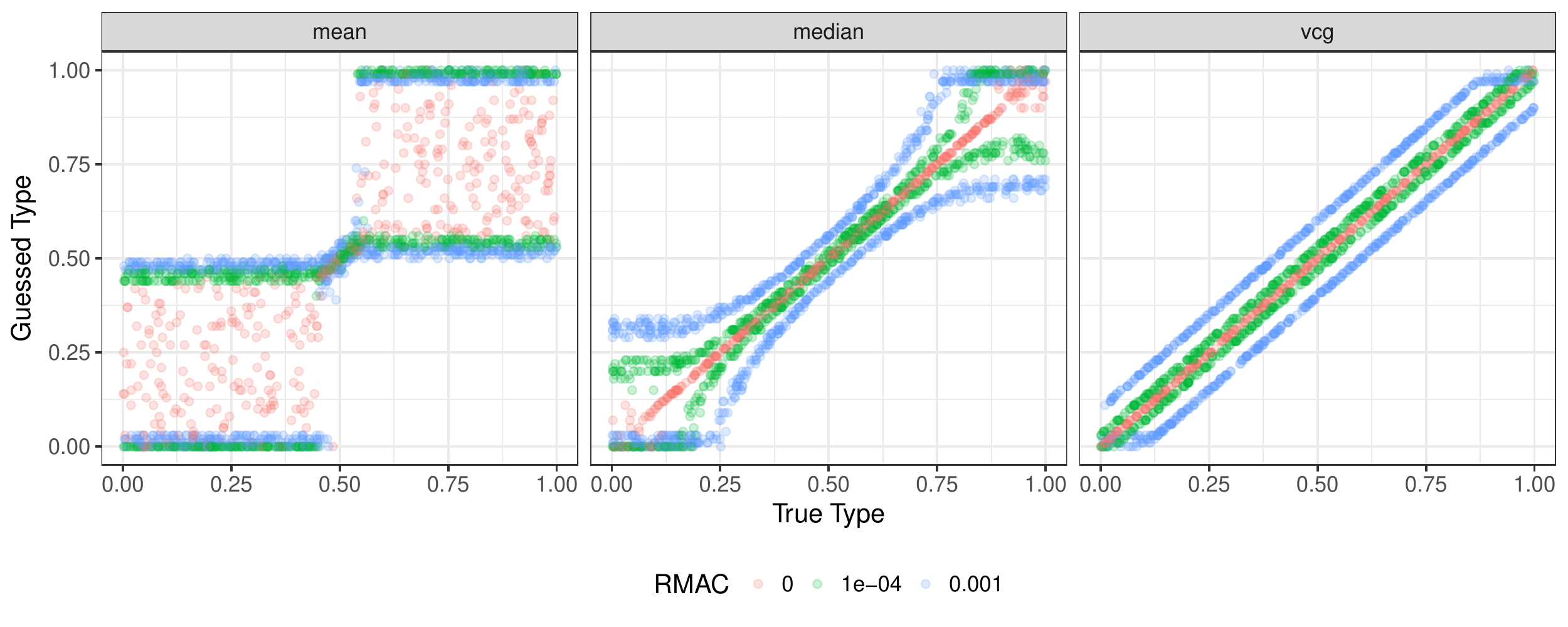}
         \caption{The median and VCG mechanisms are dominant strategy truthful but they have very different robustness properties. The mean mechanism is not well identified as types outside a narrow interval all report extreme values in equilibrium, however, RMAC bounds are defined for this case as well.}
  \label{votingexpt}
\end{figure*}

In the \textit{VCG} mechanism individuals pay the mechanism their externality on everyone else (i.e. the difference in total utility from choosing the mean that includes the report of $a_i$ and the one that excludes it) and the mean is chosen. As in auctions we discretize the types and actions with a grid of $.01$. We sample $1000$ types, calculate their optimal action in the mechanism, and run the revelation game. For the counterfactual valuation we use $V(\hat{\theta}) = \sum_{i} \theta.$ That is, we look for the most right or left shifted type distributions that are consistent with observed data.

Figure \ref{votingexpt} shows our results. First, we can see that with $\epsilon=0$ the mean mechanism is not identified since in equilibrium a whole range types choose the $0$ and $1$ actions. However, even with small $\mid \epsilon \mid$ the solution becomes unique. 

Even though both the median and VCG mechanisms are dominant strategy truthful they have very different robustness properties. In the median mechanism deviation from truthful reporting, in particular for extreme types, is not very costly as it can only affect outcomes if that person is pivotal. On the other hand, in the VCG any deviation also changes the price one has to pay into the mechanism, thus changing the way types can deviate under RMAC.

\section{Conclusion}
Structural estimation is an important area of counterfactual prediction. We have introduced RMAC as a way of dealing with situations where the standard structural assumptions of specification, equilibrium, point identification are not met. We have use the revelation game as an estimator for counterfactual quantities and adapted standard fictitious play to solve for pessimistic and optimistic equilibria of the revelation game. 

There are many possible extensions to our work. We have assumed that deviations from optimal behavior can be arbitrary but must incur low regret. It is well known that in many situations deviations from rational behavior are not random but rather systematic and can even incur large regret \citep{camerer2011advances}. An important extension of our work is to incorporate theoretical models from behavioral economics into RMAC predictions \citep{fudenberg2006dual,ragain2016pairwise,peysakhovich2017learning,peysakhovich2018reinforcement}.

Our method applies a vanilla version of fictitious play but it is well known that modifications to standard learning algorithms can lead to large changes in real world performance, especially in multi-agent settings \citep{conitzer2007awesome,syrgkanis2015fast,kroer2015faster}. Thus, it is worth exploring the use of other algorithms than RFP to solve for RMAC bounds. In addition, we assume access to a tabular and discrete representation of the counterfactual game and an important future direction is to expand these ideas to more complex multi-agent environments, for example those including multiple steps and planning \citep{shu2018m}. Such extensions would naturally require multi-agent learning algorithms that can handle function approximation such as those based on deep learning \citep{heinrich2016deep,dutting2017optimal,lowe2017multi,feng2018deep,brown2018deep}. 

We have looked at predicting counterfactual behavior in the kinds of environments studied by market designers. However, the general question of how other agents would respond to something (e.g. a behavior or change in environment) is an important problem for agent design and in particular learning whether a particular partner (or partners) are attempting to cooperate or compete \citep{littman2001friend,kleiman2016coordinate,lerer2017maintaining,shum2019theory}. Expanding RMAC to such situations is another important future direction.

\bibliography{20190325_rlmech.bbl}
%\bibliography{rlmd.bib}
\bibliographystyle{ACM-Reference-Format}

\section{Appendix}

\subsection{A Mathematical Program for the General Revelation Game}
We now present a mathematical program for solving the revelation game exactly for small instances. Throughout we will treat $V$ as a black box, assumed to be representable in the same class as the mathematical program it is stated within. Similarly we will assume that the $\text{Regret}$ functions are representable within the given class. If these assumptions are not true then the problem will of course be harder than the stated class of mathematical programs. 

Throughout the section we will abuse notation slightly in the name of readability and say that $u(a,a_{-j},\theta_j) = \mathbb{E}_{\tilde a \sim a_{-j}}[u^{\mathcal{G}'}(a,\tilde a,\theta_j)]$, i.e. the expected utility of action $a_j$ given the distribution over actions taken by other players in $\mathcal{G}'$ given the action assignment of the data-players.

First, we give a mathematical program for solving the general case of the revelation game. Here we let $a$ and $\theta$ be vectors of action and type choices, since this formulation is guaranteed to have a pure-strategy BNE:

\begin{equation}
\begin{array}{rrclll}
  \min_{\theta, a} & V(\theta, a) & \\
  \textrm{s.t.} & \max_{a \in \mathcal{A}} u(a,a_{-j},\theta_j) - u(a_j,a_{-j},\theta_j) & \leq & \epsilon   & \forall j \in \mathcal{D}  \\
  & \text{Regret}^{\mathcal{G}}_j(\theta_j) &\leq& \epsilon & \forall j \in \mathcal{D} \\
  & \multicolumn{3}{c}{\theta_j \in \Theta, a_j \in \mathcal{A}} & \forall j\in \mathcal{D}, a\in \mathcal{A}, \theta \in \Theta
\end{array}
\label{mpec}
\end{equation}

The first constraint in \eqref{mpec} is an equilibrium constraint over $\mathcal{G}'$, and therefore the general problem is a \emph{mathematical program with equilibrium constraints} (MPEC). Thus the general program is quite hard. If we make the assumptions that $\mathcal{A}$ and $\Theta$ are nonempty convex sets, and each $u(\cdot, a_{-j}|\theta_j)$ is a concave function in the choice of action $a_j$ then we can formulate the problem as a variational inequality problem:
\begin{equation}
\begin{array}{rrclll}
  \min_{\theta, a} & V(\theta, a) & \\
  \textrm{s.t.} & \langle a' - a, F(a,\theta) \rangle & \leq & \epsilon  &   \\
  & \text{Regret}^{\mathcal{G}}_j(\theta_j) &\leq& \epsilon & \forall j \in \mathcal{D} \\
  & \multicolumn{3}{c}{\theta_j \in \Theta, a_j \in \mathcal{A}} & \forall j\in \mathcal{D}, a\in \mathcal{A}, \theta \in \Theta
\end{array}
\label{mpvi}
\end{equation}
where $F(a,\theta_j)_j = \nabla_{a_j} u(a,\theta_j)$ is the gradient operator of $u$ for the given choice of $\theta$.

\subsection{A Mixed Integer Program for Two Player $\mathcal{G}'$}
Next, we give a mixed integer program (MIP) for the special case where $\mathcal{G}'$ has only two players, but where we may have an arbitrary finite number of data points. Furthermore, for this MIP we assume that $\Theta$ is discrete and finite, as well as $\mathcal{A}$ is finite.

The program has a Boolean variable $T_j^\theta$ for each pair of data point $j$ and type $\theta$, indicating whether data point $j$ takes on type $\theta$. For each data point $j$ and action $a$ we have $\sigma_j(a)\in [0,1]$ indicating the probability that $j$ puts on $a$ (we could make $\sigma_j(a)$ Boolean instead in order to compute a pure-strategy solution, but pure-strategy solutions are not guaranteed to exist when types are discrete).

We also have the following $\epsilon$-BNE-enforcing variables: $v_\theta$ represents the utility achieved by type $\theta$ in $\mathcal{G}'(T)$ under the computed solution, the slack variable $s_{\theta,a}$ denotes the \emph{inoptimality} of $a$ when taken by type $\theta$, and $\delta_{\theta,a}$ is an indicator variable denoting whether $a$ is played by any data-player taking type $\theta$. The idea of the MIP is to ensure $s_{\theta,a}\leq \epsilon$, i.e. that inoptimality is bounded by $\epsilon$, whenever any data-player  chooses type $\theta$ and puts nonzero probability on $a$. 
\begin{equation}
\begin{array}{rrclll}
  \min_{T_j^\theta,  \sigma(a), v_\theta, s_{\theta,a}, \delta_{\theta,a}} & V(T, \sigma(a)) & \\
  \textrm{s.t.} &  s_{\theta,a} -  M\delta_{\theta,a}& \leq & \epsilon  & \forall a \in \mathcal{A},\theta \in \Theta  \\
  & \sigma_j(a) + T_j^\theta - \delta(\theta,a) &\leq& 1 & \forall j \in N, a \in \mathcal{A}, \theta \in \Theta\\
  &\sum_{j',a'} \sigma_{j'}(a') u(a,a',\theta) + s_{\theta,a} & = & v_\theta & \forall j \in N, a \in \mathcal{A}, \theta \in \Theta\\
  &\sum_{a\in \mathcal{A}} \sigma_j(a) & = & 1 & \forall j \in N\\
  &\sum_{\theta \in \Theta} T_j^\theta & = & 1 & \forall j \in N\\
  & \multicolumn{3}{c}{T_j^\theta, \delta_{\theta,a} \in \{0,1\}, \sigma_j(a), s_{\theta,a}\geq 0} & \forall j\in N, a\in \mathcal{A}, \theta \in \Theta
\end{array}
\end{equation}
Note that since $\Theta$ is finite we can preprocess it and remove all $\theta$ such that $\mathcal{R}^G_j(\theta)> \epsilon$, and thus we do not need to enforce this constraint on $T_j^\theta$ in the MIP.

\subsection{Proofs of Theorems}
\begin{proof}[Proof of Theorem \ref{revgame}]
The proof relies heavily on the fact that the revelation game's utility function is defined with respect to \textit{regret} \textbf{not} the \textit{original utility function}. Suppose that a data-player has true type $\theta_j$ but reports $\theta'_j.$ In revelation-game BNE this $\theta_j'$ must have zero regret. But this violates the identification assumption, since we could then construct a new distribution $\mathcal{F}'$ where we reassign type $\theta_j$ to $\theta_j'$ but keep the same distribution over actions in $\mathcal{G}$ as part of a BNE. Thus the reported distribution over types must be $\mathcal{F}$ in revelation-game BNE. Now we can use the uniqueness assumption to infer that each data-player reports their true type, as well as their action in the unique BNE of $\mathcal{G}'$ given distribution $\mathcal{F}$. If they report any other action they must have nonzero regret, or they would violate the uniqueness assumption.
\end{proof}

\begin{proof}[Proof of Theorem \ref{rfpishard}]
The first statement is by reduction from max-social-welfare Nash equilibrium in some game $G^{SW}$, which is NP-hard~\citep{conitzer2008new}. We set $\epsilon = 0$, and $V(\theta,a)$ equal to the negative social welfare in $G^{SW}$ of actions $a$. For each agent in the NE problem we instantiate a data point $d_i$ and create the game $G$ such that each $i$ can only take on the type corresponding to their payoffs in $G^{SW}$ (this is easily done by making every other type have non-zero regret in $G$). Now we set $G' = G^{SW}$. A solution to the RMAC problem now corresponds to a social-welfare maximizing Nash equilibrium of $G^{SW}$.

The second statement is by reduction from the problem of checking whether a pure-strategy BNE exists, which is NP-complete~\citep{conitzer2008new}. Consider a symmetric game $G^{pure}$ that we wish to find a pure-strategy BNE for. We let $G'=G^{pure}$. For each type $\theta$ of $G^{pure}$ we instantiate a data point such that only $\theta$ is a feasible type. Now the distribution over types in $G'$ equals that of $G^{pure}$, and so the equilibria are in correspondence.
\end{proof}

\begin{proof}[Proof of Theorem \ref{rfpconverge}]
First we show that the limit $\sigma^*$ is an $\epsilon$-BNE. Let $(\bar \theta, \bar a)$ denote a sequence of play in question. Denote by $\bar{\sigma}^t_j$ the strategy of player $j$ implied by the history $(\bar \theta, \bar a)$ up to time $t$.  % Denote by $\bar{d}_j (\bar \theta, \bar a) [\theta, a]$ be the probability this distribution puts on revelation game actions $(\theta, a).$
  
Suppose $\sigma^*$ is not an $\epsilon-$BNE. Then there exists data player $j$ and revelation game actions $(\theta_j,a_j)$ and $(\theta_j',a_j')$ that are both in the support of $\sigma^*$ but have the following payoff difference $$\mathcal{U}_j^{rev}(\theta_j, a_j, a_{-j}^*,\mathcal{D}) - \mathcal{U}_j^{rev}(\theta_j',a_j', a_{-j}^*,\mathcal{D}) > \epsilon + \epsilon'$$ for some $\epsilon'>0$. 
  
Now pick $T$ such that for all $t \geq T$ we have $$|\bar \sigma_j^t - \sigma^*_j | \max_{\theta_j, a_j, a_{-j}} \mathcal{U}_j^{rev}(\theta_j, a_j, a_{-j},\mathcal{D}) \leq \frac{\epsilon'}{2K}$$ where $K$ is the number of pure strategy profiles. Such a $T$ exists since by assumption $(\bar \theta, \bar a)$ converges and $\mathcal{U}_j^{rev}$ is bounded. We then have
  \begin{align*}
    \mathbb{E}[\mathcal{U}_j^{rev}(\theta_j', a_j', \bar a_{-j}^t,\mathcal{D})]
    &= \sum_{(\theta_{-j},a_{-j})} \mathcal{U}_j^{rev}(\theta_j', a_j', a_{-j},\mathcal{D}) \bar \sigma^t (\theta_{-j},a_{-j}) \\
    &\leq \sum_{(\theta_{-j},a_{-j})} \bigg[ \mathcal{U}_j^{rev}(\theta_j', a_j', a_{-j},\mathcal{D}) \bar \sigma^t(\theta_{-j},a_{-j}) + \frac{\epsilon'}{2K} \bigg] \\
    &\leq \sum_{(\theta_{-j},a_{-j})} \mathcal{U}_j^{rev}(\theta_j', a_j', a_{-j},\mathcal{D}) \bar \sigma^t (\theta_{-j},a_{-j}) + \frac{\epsilon'}{2} \\
    &< \sum_{(\theta_{-j},a_{-j})} \mathcal{U}_j^{rev}(\theta_j,a_j,a_{-j},\mathcal{D}) \bar \sigma^t (\theta_{-j},a_{-j}) - \frac{\epsilon'}{2} - \epsilon\\
    &\leq \sum_{(\theta_{-j},a_{-j})} \bigg[ \mathcal{U}_j^{rev}(\theta_j,a_j,a_{-j},\mathcal{D}) \bar \sigma^t (\theta_{-j},a_{-j}) + \frac{\epsilon'}{2K}\bigg] - \frac{\epsilon'}{2} - \epsilon\\
    &\leq \sum_{(\theta_{-j},a_{-j})} \mathcal{U}_j^{rev}(\theta_j,a_j,a_{-j},\mathcal{D}) \bar \sigma^t (\theta_{-j},a_{-j}) - \epsilon\\
    &= \mathbb{E}[\mathcal{U}_j^{rev}(\theta_j, a_j, \bar a_{-j}^t,\mathcal{D})] - \epsilon
  \end{align*}
Thus we have that after iteration $T$ we no longer select $(\theta_j',a_j')$ since it is not within the set of $\epsilon$ best responses. This follow from the above algebra and the fact that $\mathcal{U}_j^{rev}$ is bounded above by zero (since it is the negative maximum regret).

But this implies that thus $$\lim_{t \to \infty} \bar \sigma^t (\theta_{-j},a_{-j}) \rightarrow 0.$$ But this is a contradiction since we assumed $\sigma^* (\theta'_{j}, a'_{j}) > 0.$

Now we prove local optimality. 

Suppose we do not have local optimality. Then there exists $j$ and revelation game actions $(\theta_j, a_j)$ such that $$\mathbb{E}[\mathcal{U}_j^{rev}(\theta_j, a_j, a^*_{-j},\mathcal{D})] + \epsilon' < \epsilon$$ and $$V(\theta_j, a_j, \sigma^*_{-j}) + \epsilon' < V(\sigma^*)$$ for some $\epsilon'>0$.  

Since the expected value of $V$ is continuous in the empirical distribution there exists $(\theta_j', a_j')$ with $\sigma^* (\theta_j',a_j')>0$ such that $V(\theta_j, \bar\theta_{-j}, a_j, \bar a_{-j}^t) + \epsilon'' < V(\theta_j',\bar\theta_{-j},a_j',\bar a_{-j}^t)$ for all $t \geq T'$ for some sufficiently large $T'$.
 
Now pick $T\geq T'$ such that for all $t\geq T$, $(\theta_j,a_j)$ is in the $\epsilon$-best-response set to $\sigma^*_{-j}$ for $j$. Such a $T$ is guaranteed to exist by continuity of $V$ and $\mathcal{U}_j^{rev}$ in the empirical distribution. But then best responses never select $(\theta_j', a_j')$ after $T$ which is a contradiction.
\end{proof}

\end{document}